\documentclass[12pt]{iopart}

\usepackage[utf8]{inputenc}
\usepackage[english]{babel}
\usepackage[T1]{fontenc}
\usepackage{hyperref}
\usepackage{bbold}
\usepackage{color}
\usepackage{comment}
\usepackage{braket}
\usepackage{graphicx}
\usepackage{subfig}
\usepackage{tikz}
\usepackage{lipsum}
\usepackage{lscape}
\usepackage{amsthm}
\usepackage{amsopn}
\usepackage{relsize}
\usepackage[backend=bibtex8,                
maxnames=6,                         
bibstyle=numeric,               
hyperref=true,                  
citestyle=numeric-comp, 
sorting=none]{biblatex}

\newcommand{\euler}{e}
\newcommand{\ramuno}{i}
\newcommand{\qudit}[1]{\left\vert #1 \right\rangle}
\newcommand{\rqudit}[1]{\left\langle #1 \right\vert}

\newcommand{\Z}{\mathbb{Z}}
\newcommand{\C}{\mathbb{C}}
\newcommand{\R}{\mathbb{R}}
\DeclareMathOperator{\mod}{mod}

\newtheorem{thm}{Theorem}
\newtheorem{theorem}[thm]{Theorem}

\newtheorem{definition}[thm]{Definition}

\bibliography{mybibliography}

\begin{document}

\title{On the convex characterisation of the set of unital quantum channels}

\author{Constantino Rodriguez-Ramos and Colin M.\ Wilmott}

\address{Department of Mathematics, Nottingham Trent University, Clifton Campus Nottingham NG11 8NS, UK}
\ead{tinorodram@gmail.com}
\vspace{10pt}
\begin{indented}
\item[]March 2023
\end{indented}

\begin{abstract}
In this paper, we consider the convex structure of the set of unital quantum channels. To do this, we introduce a novel framework to construct and characterise different families of low-rank unital quantum maps. In this framework, unital quantum maps are represented as a set of complex parameters on which we impose a set of constraints. The different families of unital maps are obtained by mapping those parameters into the operator representation of a quantum map. For these families, we also introduce a scalar measuring their distance to the set of mixed-unitary maps. We consider the particular case of qutrit channels which is the smallest set of maps for which the existence of non-unitary extremal maps is known. In this setting, we show how our framework generalises the description of well-known maps such as the antisymmetric Werner-Holevo map but also novel families of qutrit maps.
\end{abstract}

%
%
%
%
%

\section{Introduction}

Quantum channels provide the most general characterisation of the arbitrary evolution of quantum systems and  represent a vital ingredient in establishing quantum computing and communication. For instance, in quantum key distribution protocols, the amount of overall noise in the quantum channel determines the rate at which secret bits are distributed between authorised parties. Mathematically, quantum channels are characterised as completely positive trace-preserving (CPT) linear mappings between density matrices. An interesting class of these maps are the unital completely positive trace-preserving (UCPT) maps which are those quantum channels sending the noisiest state of the system, the maximally-mixed state, to itself. There are several good reasons to consider unital quantum maps instead of general quantum channels. For low-dimensional systems, the additional constraint required by unitality often simplifies problems and allows for a geometrical intuition of the state space. Some crucial advances in quantum information theory like the parametrization of qubit channels are closely connected to the prior characterization of the structure of the set of unital maps \cite{Ruskai2002, Fujiwara1999}. 

A particular feature of qubit unital channels is that they always admit a decomposition in terms of convex combinations of unitary channels \cite{Kummerer1987}. This property allows us to associate the set of UCPT maps with the geometry of a 3-simplex in which the vertices correspond to the set of Pauli channels \cite{bengtsson2017geometry}. However, this property of qubit unital channels is no longer true for unital channels of higher dimension. Various investigations introduced examples of maps which were neither unitary channels nor could be decomposed in terms of convex combinations of unitary channels \cite{Tregub1986, Landau1993}. The existence of unital maps which are not mixed-unitary establishes a crucial difference between qubit channels and higher dimensional channels. From the point of view of convex decomposition, the description of the set of unital qubit maps is analogous to the description of doubly-stochastic matrices given by Birkhoff's theorem. This theorem establishes that doubly-stochastic matrices can be decomposed in terms of a convex combination of permutation matrices. Interestingly, the existence of non-unitary extremal maps within the set of UCPT maps implies that Birkhoff's theorem cannot be extended to other maps of higher dimension. 

Previous works investigated the convex structure of the set of unital maps and their relation to the set of mixed-unitary maps. For example, Audenaert et al. considered the distance between these two sets \cite{Audenaert2008}, and Mendl and Wolf provided computable criteria for the separation of the unital channels from the mixed-unitary set \cite{Mendl2008}. In this work, we develop tools for the study of the structure of the set of unital quantum maps. In particular, we provide a novel framework to construct families of maps with a fixed Kraus rank. We also provide a computable measure quantifying the relation of unital maps we introduce and the set of convex combinations of unitary maps. For the particular case of qutrit maps, we find that the families of maps we construct can generalise several well-known maps appearing in the literature such as the Weyl maps or the anti-symmetric Werner-Holevo map \cite{Werner2002}. 

This document is organised as follows. In section \ref{prelim}, we outline the necessary tools required for the study of UCPT maps and their convex structure. In section \ref{family}, we consider a parametrised family of quantum maps and derive the constraints that guarantee that such maps are unital and trace-preserving in terms of the defining parameters of the family. The same parametrisation was considered in the context of circuit decomposition \cite{Wang2015}. Our contribution is to use this parametrised family of maps to investigate the geometry of the set of UCPT maps. In addition, we provide a generalisation of this family of maps to consider UCPT maps with different ranks. We also construct a scalar determining the distance of the family of maps defined and the set of convex combinations of mixed-unitary maps. In section \ref{examples}, we consider the particular case of qutrit maps. We show that the framework introduced describes well-known UCTP-extremal maps such as the Heisenberg-Weyl maps and the antisymmetric Werner-Holevo channel \cite{Werner2002}.

\section{Preliminaries}\label{prelim}

Let $n \in \mathbb{N}$, the possible states of a $n$-dimensional quantum system are represented as vectors $\ket{\psi}\in\mathcal{H}_{n}$ where $\mathcal{H}_{n}$ represents a Hilbert space of dimension $n$ equipped with the particular inner product $\braket{\psi_A,\psi_B}\in\mathbb{C}$, which is anti-linear in the first argument and linear in the second. Linear operators acting over the system are represented as $n\times n$ complex matrices $O\in\mathcal{M}_{n}$. For example, we write $\mathbb{1}_n\in\mathcal{M}_{n}$ to denote the identity operator on $\mathcal{H}_n$. The space of linear operators $\mathcal{M}_{n}$ can be regarded as a Hilbert space when equipped with the Hilbert-Schmidt inner product defined as
\begin{equation}
\braket{O_A,O_B}_{\mathcal{HS}}=\Tr{O^{\dagger}_A O_A}.
\end{equation}
Open systems are described in terms of ensembles of quantum states $\{\ket{\psi_i}\}_{i\in\mathbb{Z}}$ for which each state $\ket{\psi_i}$ has an associated probability $p_i\in\mathbb{R}^{+}$ and $\sum_i{}p_i=1$. State ensembles can be described by a, so-called, density operator $\rho$ corresponding to a hermitian, positive positive semi-definite matrix with unit trace. We denote by $\mathcal{D}_n \subset \mathcal{M}_n$ to the space of all density matrices. 

Quantum operations represent the possible transformations that a quantum system may undergo. Quantum operations correspond to linear mappings between spaces of density matrices and consequently have certain constraints which we describe as follows. The mapping $\mathcal{E}: \mathcal{D}_n \mapsto \mathcal{D}_m$ is said to be positive if it sends positive semi-definite matrices to positive semi-definite matrices, and completely positive if ${\mathcal{E}} \otimes \mathbb{1}_n$ is positive for all $n$. A completely positive mapping $\mathcal{E}$ is trace-preserving if and only if $\Tr{\mathcal{E}(\rho)}=\Tr{\rho}$ for all $\rho \in \mathcal{D}_{n}$. Following the fact that a quantum operation is a mapping between density matrices, quantum operations correspond necessarily to complete positive trace-preserving (CPT) maps.
We denote the set of all CPT maps by $\Xi^{T}_{n,m}$. In this work, we consider the structure set of CPT maps leaving the maximally mixed state invariant $\mathcal{E}(\mathbb{1}_n/n) =\mathbb{1}_m/m$ which are called unital maps. We represent the set of all CP maps which are only unital by $\Xi^{U}_{n,m}$ and the set of all unital complete positive trace-preserving (UCPT) maps by $\Xi^{UT}_{n,m}:=\Xi^{T}_{n,m}\bigcap\Xi^{U}_{n,m}$. 

\subsection{Operator-sum representation of a map}

The map $\mathcal{E}: \mathcal{D}_n \mapsto \mathcal{D}_m$ is completely positive if and only if it admits a representation of the form  
\begin{eqnarray}
\mathcal{E}(\rho)=\sum_{i=1}^{r}K_i\rho K_i^{\dagger}
\label{OSR}
\end{eqnarray} 
for all $\rho \in \mathcal{D}_n$, where the matrices $K_i \in \mathcal{M}_{m\times n}$  are referred to as the Kraus operators \autocite{Choi1975}. This form of expressing a map is known as the operator-sum representation. A complete positive map $\mathcal{E}: \mathcal{D}_n \mapsto \mathcal{D}_m$ given by \eref{OSR} is trace-preserving if
\begin{eqnarray}
\sum_{i=1}^{r}K_i^{\dagger}K_i = \mathbb{1}_n
\label{TP}
\end{eqnarray}
and $\mathcal{E}$ is unital if
\begin{eqnarray}
\sum_{i=1}^{r}K_iK_i^{\dagger} = \mathbb{1}_m.
\label{UN}
\end{eqnarray} 

The operator-sum representation of a map is not unique. The following theorem establishes when two sets of operators represent the same map \autocite{nielsen2002quantum}.

\begin{theorem}
Suppose \(\{K_i\}_{i=0,\dots,n}\) and \(\{G_j\}_{j=0,\dots,n}\) are the sets of Kraus operators defining the CP maps $\mathcal{E}$ and $\mathcal{F}$, respectively. Then $\mathcal{E}=\mathcal{F}$ if and only if there exist complex numbers \(u_{ij}\) such that \(K_i=\sum_j{u_{ij}G_j}\) and \(U=(u_{ij})_{i,j\in\Z_n}\) is an \(n\) by \(n\) unitary matrix.
\label{UF}
\end{theorem} Two sets of Kraus operators with different cardinality represent the same map if by appending zero operators to the set with fewer elements, the unitary freedom condition is satisfied. Consider the map given in terms of the operators $\{K_i\}_{i\in\Z_n}$. Choi showed that such representation is minimal if and only if the operators $K_i$ are linearly independent \autocite{Choi1975}. We may define the Kraus rank as the cardinality of the minimal representation of the map.

\subsection{Choi representation of a map}

Consider now the representation of the map $\mathcal{E}: \mathcal{D}_d \mapsto \mathcal{D}_d$ in terms of the Choi operator $C_{\mathcal{E}}$ given by
\begin{equation}
	C_{\mathcal{E}}=(\Phi \otimes \mathbb{1}_d)\ket{\psi} \bra{\psi}
\end{equation}
where \(\ket{\psi}\) is a maximally entangled pure state i.e. $\ket{\psi}=\sum_{m=0}^{d-1}\ket{m}\ket{m}$. The operator-sum representation of $\mathcal{E}$ can always be recovered from the Choi state. In particular, the set of $r$ eigenvectors of $C_{\mathcal{E}}$
(in matrix form) multiplied with their respective eigenvalues is a valid Kraus set $\{K_i\}_{i\in\Z_r}$ defining $\mathcal{E}$ as in \ref{OSR} \cite{bengtsson2017geometry}. In this document, we will use the two introduced map representations depending on which is more useful in the particular problem.


\subsection{Convex characterization}

The set of trace-preserving maps, $\Xi^{T}_{n,m}$, and its adjoint, the set of unital maps, $\Xi^{U}_{n,m}$, are convex. So all possible maps of the form
\begin{equation}
\mathcal{E}_{AB}(p)=\mathcal{E}_A + (1-p)\mathcal{E}_B.
\end{equation}
where $0<p<1$ and $\mathcal{E}_A, \mathcal{E}_B \in \Xi^{T}_{n,m}(\Xi^{U}_{n,m})$ and  are trace preserving (unital). The elements of a set which do not admit such decomposition are called extreme points of the set. The concise characterisation of extreme points of $\Xi^{U}_{n,m}$ was provided by Choi in the following theorem \autocite{Choi1975}.
\begin{theorem}
Consider the set of UCP maps $\mathcal{E}: \mathcal{D}_n \mapsto \mathcal{D}_m$ with minimal operator-sum representation $\mathcal{E}(\rho)=\sum_{i=1}^{r}K_i\rho K_i^{\dagger}$. Then, $\mathcal{E}$ is an extreme point of $\Xi^{U}_{n,m}$ if and only if the set $\{K_i K^\dagger_j\}_{i,j \in \Z_r}$ is linearly independent.
\label{THUN}
\end{theorem}

Choi's theorem has a natural extension provided that the set of CPT maps is the dual of the set of UCP maps with respect to the complex conjugation. The following theorem establishes when a CPT map is an extreme point of the set $\Xi^{T}_{n,m}$.

\begin{theorem}
Consider the set of CPT maps $\mathcal{E}: \mathcal{D}_n \mapsto \mathcal{D}_m$ with minimal operator-sum representation $\mathcal{E}(\rho)=\sum_{i=1}^{r}K_i\rho K_i^{\dagger}$. Then, $\mathcal{E}$  is an extreme point of $\Xi^{T}_{n,m}$ if and only if the set $\{K^\dagger_i K_j\}_{i,j \in \Z_r}$ is linearly independent.
\label{THTP}
\end{theorem}

Theorem \ref{THUN} and theorem \ref{THTP} establish bounds to the Kraus rank of the extreme points of the set of unital maps and the set of trace-preserving maps, respectively. The Kraus rank of an extreme point of $\Xi^{U}_{n,m}$ is upper bounded by \(m\). This follows from the fact that at most \(m^2\) matrices $K_i K^\dagger_j\in\mathcal{M}_{m\times m}$ can be linearly independent For the CPT case, we have it that the Kraus rank of an extreme point of $\Xi^{T}_{n,m}$ is upper bounded by \(n\) as at most \(n^2\) matrices $K^\dagger_i K_j\in\mathcal{M}_{n\times n}$ can be linearly independent. The set $\Xi^{UT}_{d,d}$ is also convex and the following theorem originally stated in \autocite{Landau1993} characterises its extreme points.
\begin{theorem}
\label{LS}
Consider the set of UCPT maps \(\mathcal{E}: \mathcal{D}_{d} \rightarrow \mathcal{D}_{d}\) where $\mathcal{E}(\rho):=\sum_{i=1}^{r}K_i\rho K_i^{\dagger}$  and  $\sum_{i=1}^{r}K_iK_i^{\dagger} =\sum_{i=1}^{r}K_i^{\dagger}K_i = \mathbb{1}_d$. Then, $\mathcal{E}$ is an extreme point of $\Xi^{UT}_{n,m}$ if and only if the set of \(2d \times 2d\) matrices
\begin{equation}
\label{extEQN}
    \{K^{\dagger}_i K_j \oplus K_i K^{\dagger}_j\}_{i,j\in\Z_r}
\end{equation}
is linearly independent.
\end{theorem}

From this theorem, it follows that the Kraus rank of an extreme point of $\Xi^{UT}_{d,d}$ is upper bounded by $\sqrt{2d^2}$ since at most $2d^2$ matrices $K^{\dagger}_i K_j \oplus K_i K^{\dagger}_j$ can be linearly independent as their number of non-zero elements of these matrices is bounded by $2d^2$. For example, for dimension, $d=3$ the rank of the extreme points of unital and trace-preserving maps is upper bounded by $\sqrt{18}$.

\section{Parametrised UCPT maps}\label{family}

Consider the family of CP maps over dimension $d$, $\mathcal{E}: \mathcal{D}_d \mapsto \mathcal{D}_d$ acting on a density matrix $\rho$ as
\begin{eqnarray}
\mathcal{E}(\rho)=\sum_{i=0}^{d-1}K_i\rho K_i^{\dagger}\quad \textrm{ with }\nonumber\\
K_i = \sum_{j}^{d-1}{}\alpha_{ij}X_iZ_j,\quad \alpha_{ij}\in \C
\label{dfam1}
\end{eqnarray}
where $\{X_i\}_{i\in\Z_d}$ and $\{Z_i\}_{i\in\Z_d}$ are the shift and clock matrices, respectively. The shift and clock matrices are expressed in terms of Dirac notation as $X_i=\sum_{k=0}^{d-1}{}\ket{k+i}\bra{k}$ and $Z_i=\sum_{k=0}^{d-1}{}\omega^{ik}\ket{k}\bra{k}$ with $\omega=\euler^{\frac{2\pi}{d}\ramuno}$. 

The properties of the maps given by \eref{dfam1} derive from the properties of the shift and clock matrices. Following the orthogonality relation between two different shift matrices $\braket{X_i,X_j}_{\mathcal{HS}}=\delta_{ij}$, we have that the Kraus operators $\{K_i\}_{i\in\Z_d}$ in \eref{dfam1} are also orthogonal between them. Consequently, by Theorem \ref{UF}, the representation of a map $\mathcal{E}$ is unique except for a global phase applied to each one of the Kraus operators. Shift matrices $X_i$ determine the positions of the non-zero elements in the Kraus representation while clock matrices $Z_i$ along with $\alpha_{ij}$ are used to span all possible non-zero values. This follows from the fact that combinations given by $\alpha_{i;0}Z_0+\dots+\alpha_{i;d-1}Z_{d-1}$ span all diagonal matrices. If we substitute the explicit expression of $X_i$ and $Z_i$ in the Kraus operators as given by \eref{dfam1}, we obtain that
\begin{equation}
K_i =\sum_{j,k=0}^{d-1}{}\alpha_{ij}\omega^{kj}\qudit{k+i}\rqudit{k}.
\label{dfam3}
\end{equation}
We can see that each $K_i$ is determined by $d$ complex coefficients $\alpha_{i0},\dots,\alpha_{id}$ corresponding to $2d$ real parameters. However, the multiplication of any operator with a global phase does not affect the representation of the map. We obtain a univocal representation for each $K_i$ if we restrict one of the coefficients $\alpha_{ij}$ to the real set. For example, if we set $\alpha_{i0}\in\R$, each Kraus operator has a unique representation in terms of $2d-1$ real parameters and the family maps given by \eref{dfam1} is described in terms of $2d^2-d$ real parameters. 

At this point, we may consider the properties of the different maps given by \eref{dfam1} in terms of $\alpha_{ij}$. The following theorem establishes the conditions required on $\mathcal{E}$ to be unital and trace-preserving.
\begin{theorem}  
The map $\mathcal{E}: \mathcal{D}_d \mapsto \mathcal{D}_d$ as given by \eref{dfam1} is trace-preserving if
\begin{equation}
\sum_{i,j=0}^{d-1}{\alpha_{ij}\alpha^*_{ij}}=1 \label{c1}\\
\end{equation}
and
\begin{equation}
\sum_{i,j=0}^{d-1}{\alpha_{ij+l}\alpha^*_{ij}}=0,\quad l=1,\dots,d-1.\label{c2}
\end{equation}
The map $\mathcal{E}$ is unital if in addition to condition (\ref{c1}), we have it that 
\begin{equation} 
\sum_{i,j=0}^{d-1}{\alpha_{ij+l}\alpha^*_{ij}\omega^{-il}}=0 \quad l=1,\dots,d-1.
\label{c3}
\end{equation}
\label{UTConditions}
\end{theorem}

\begin{proof} Let us consider  the set  $\{K_i^\dagger K_i\}_{i \in \Z_d }$ in the $\{\qudit{a}\rqudit{b}, \ a,b \in \Z_{d}\}$ basis as 
\begin{eqnarray}
K_i^\dagger K_i &= \left(\sum_{k,j=0}^{d-1}{}\alpha^*_{ij}\omega^{-kj}\qudit{k}\rqudit{k+i}\right)\left(
\sum_{m,n=0}^{d-1}{}\alpha_{im}\omega^{mn}\qudit{n+i}\rqudit{n}\right)\nonumber\\
  &=\sum_{k,j,m,n=0}^{d-1}\alpha_{im}\alpha^*_{ij}\omega^{mn-kj}\ket{k}\braket{k+i|n+i}\bra{n}\nonumber\\
  &=\sum_{k,j,m=0}^{d-1}\alpha_{im}\alpha^*_{ij}\omega^{k(m-j)}\ket{k}\bra{k}.
   \label{KTK}
\end{eqnarray}
To satisfy the trace-preserving condition \(\sum_{i=0}^{d-1}K^{\dagger}_i K_i=\mathbb{1}_d\), 
it necessarily follows that
\begin{equation}
\sum_{i,j,m=0}^{d-1}\alpha_{im}\alpha^*_{ij}\omega^{k(m-j)}=1 \textrm{  for  } k=0,\dots,d-1.
\label{close1}
\end{equation}
\noindent
By the change of index, $m-j=l$, \eref{close1} can be expressed as
\begin{equation}
\sum_{i,l=0}^{d-1}{}\sum_{j=0}^{d-1}{}\alpha_{ij+l}\alpha^*_{ij}\omega^{kl}=1 \textrm{  for  } k=0,\dots,d-1
\end{equation}
\noindent
and using the change of variable $\beta_{il}=\sum_{j=0}^{d-1}{}\alpha_{ij+l}\alpha^*_{ij}$, we get that

\begin{equation}
\sum_{i,l=0}^{d-1}{}\beta_{il}\omega^{kl}=1 \textrm{  for  } k=0,\dots,d-1.
\end{equation}
The unique solution to this system of $d$ linearly independent equations in terms of the set of variables $\{\beta_{il}\}_{i,l\in\Z_d}$ corresponds to $\sum_{i=0}{}\beta_{i0}=1$ and $\sum_{i=0}{}\beta_{il}=0$ for $l=1,\dots d-1$. By expressing the solution of the system in terms of the original variables $\{\alpha_{ij}\}_{i,k\in\Z_d}$ we get precisely the equations (\ref{c1}) and (\ref{c2}). 

Similarly, we can obtain the conditions required by a map to be unital. Let us consider the set  $\{K_i K_i^\dagger\}_{i \in \Z_d} $ as

\begin{eqnarray}
K_i K_i^\dagger &= \left(
\sum_{m,n=0}^{d-1}{}\alpha_{im}\omega^{mn}\qudit{n+i}\rqudit{n}\right)\left(\sum_{k,j=0}^{d-1}{}\alpha^*_{ij}\omega^{-kj}\qudit{k}\rqudit{k+i}\right)\nonumber\\
  &=\sum_{k,j,m,n=0}^{d-1}\alpha_{im}\alpha^*_{ij}\omega^{mn-kj}\ket{k+i}\braket{k|n}\bra{n+i}\nonumber\\
   &=\sum_{k,j,m,n=0}^{d-1}\alpha_{im}\alpha^*_{ij}\omega^{mn-kj}\delta_{k,n}\ket{k+i}\bra{n+i}\nonumber\\
   &=\sum_{k,j,m=0}^{d-1}\alpha_{im}\alpha^*_{ij}\omega^{(k-i)(m-j)}\ket{k}\bra{k}.
   \label{KKT}
\end{eqnarray}
To satisfy the unital condition \(\sum_{i=0}^{d-1}K_i K^{\dagger}_i=\mathbb{1}_d\), it follows that
\begin{equation}
\sum_{i,j,m=0}^{d-1} \alpha_{im}\alpha^*_{ij}\omega^{(k-i)(m-j)}=1 \textrm{  for  }k=0,\dots,d-1.
\label{close2}
\end{equation}
By the change of index, $m-j=l$, \eref{close2} can be written as
\begin{equation}
\sum_{i,l=0}^{d-1}{}\sum_{j=0}^{d-1}{}\alpha_{ij+l}\alpha^*_{ij}\omega^{(k-i)l}=1 \textrm{  for  } k=0,\dots,d-1
\end{equation}
\noindent
and using now the change of variable $\beta_{il}=\sum_{j=0}^{d-1}{}\alpha_{ij+l}\alpha^*_{ij}$ we get that

\begin{equation}
\sum_{i,l=0}^{d-1}{}\beta_{ik}\omega^{(k-i)l}=1 \textrm{  for  } k=0,\dots,d-1.
\label{syst2}
\end{equation}
We get again a system of $d$ equations in terms of $\{\beta_{ik}\}_{i,k\in\Z_d}$. The solution of this system is given by $\sum_{i,l=0}{}\beta_{i0}=1$ and $\sum_{i}{}\beta_{il}\omega^{-l}=0$ for $l=1,\dots d-1$. If we express the solution of the system in terms of the elements of the set $\{\alpha_{ij}\}_{i,j\in\Z_d}$ we get precisely the equations (\ref{c1}) and (\ref{c3}) which completes the proof.
\end{proof}

\subsection{A general framework for constructing UCTP maps}

Equations \eref{c1}, \eref{c2} and \eref{c3} provide the necessary conditions required by the maps given in \eref{dfam1} to be unital and trace-preserving. Therefore, we can associate each one of those maps with an element of the following set of complex matrices
\begin{eqnarray}
\mathcal{A}_d=(\alpha_{ij})_{i,j\in\Z_d}\in\C^{d\times d} \textrm{ with } \nonumber \\
\sum_{i,j=0}^{d-1}{\alpha_{ij}\alpha^*_{ij}}=1 \textrm{ and }
\left\{\begin{array}{c}
\sum_{i,j=0}^{d-1}{\alpha_{ij+l}\alpha^*_{ij}}=0 \\ 
\sum_{i,j=0}^{d-1}{\alpha_{ij+l}\alpha^*_{ij}\omega^{-il}}=0
\end{array}\right. \textrm{ for } l=1,...,d-1
\label{matparam}
\end{eqnarray}
where $\omega=\euler^{\frac{2\pi}{d}\ramuno}$. Conversely, we can always find a map $\pi_{d}$ that sends the complex matrices $(\alpha_{ij})_{i,j\in\Z_d}\in\mathcal{A}_d$ to UCPT maps as
\begin{equation}
\pi_{d}(\alpha_{ij})=\mathcal{E}
\end{equation}
where the map $\mathcal{E}$ is defined as in \eref{dfam1}. At this point, we wish to generalise $\pi_{d}$ to consider other UCPT maps. In particular, we are interested in finding families of maps with different Kraus ranks.

To construct such maps we establish first some definitions. Let $A=\{P_m\}_{m\in\Z}$ denote the set of ordered pairs $P_m=(a_m,b_m)\in \Z_d \times \Z_d$ such that $a_m\neq a_{m'}$ and $b_m\neq b_{m'}$ if $m\neq m'$. Now consider the map $\mathcal{P}(a,b)=a+d(b\mod d)$ which applied on $P$ sends each pair to $\mathcal{P}(P):\Z_d \times \Z_d \rightarrow \Z_{d\times d}$. Given all these elements, we can always find the following set 
\begin{equation}
S_r=\{A_n\}_{n\in\Z_{r}}
\label{setS}
\end{equation} 
such that $\bigcup_{n=0}^{r-1}\mathcal{P}(A_n)=\Z_{d\times d}$. In other words, $S_r$ is a covering set of $\Z_{d\times d}$. At this point, we can introduce a generalisation of $\pi_{d}$. To do that, we define a new map sending complex matrices $(\alpha_{ij})_{i,j\in\Z_d}\in\mathcal{A}_d$ as given by \eref{matparam} to UCTP maps $\mathcal{E}: \mathcal{D}_d \mapsto \mathcal{D}_d$ of rank $r$ which in this case can be different from $d$. We construct this new map by determining how the parameters $\alpha_{ij}$ are mapped into a Choi operator representing $\mathcal{E}$.

\begin{definition}
Let $r,d\in\mathbb{N}$ such that $r\geq d$ and let $S_{r}$ be a covering set of $\Z_{d\times d}$ as given by \eref{setS}. Then, we define the map $\pi_{r}(\alpha_{ij})$ with $\alpha_{ij}\in \mathcal{A}_d$ in terms of its Choi representation by the operator
\begin{equation}
C_\mathcal{E}=\sum_{n=0}^{r-1}\sum_{(k,h)\in A_n \atop (l,i)\in A_n}\left(\sum_{j=0}^{d-1}{\alpha_{hj}\omega^{jk}}\right)\left(\sum_{j=0}^{d-1}{\alpha^*_{ij}\omega^{-jl}}\right)\ket{\mathcal{P}(k,k+h)}\bra{\mathcal{P}(l,l+i)}.
\label{GenerChoi}
\end{equation}
\end{definition} This definition generalises the family of UCTP maps given in \eref{dfam1} as it includes other families of UCTP maps with different ranks. To recover the original family of maps, we just need to fix $r=d$ in \eref{GenerChoi} to recover the same UCTP maps. The following theorem quantifies the dimensionality of the families of UCTP maps given by $\pi_r(\mathcal{A}_d)$.

\begin{theorem}
The family of UCTP maps represented by $\pi_r(\mathcal{A}_d)$ has dimension $2d^2-2d-r+1$.
\end{theorem}
\begin{proof}
The dimension set of UCTP maps $\pi_k(\mathcal{A}_d)$ is determined by the real dimensions of the set $\mathcal{A}_d$. However, we need to consider also the arbitrariness in the choice of $\alpha_{i,j}$ which is induced by the freedom in the choice of Kraus operators. As we saw, in the case of orthogonal Kraus operators, this freedom corresponds to multiplying each Kraus operator by a complex phase. Consequently, for the map $\pi_r(\mathcal{A}_d)$ with rank $r$ such freedom is given by the action of the group $\bigotimes_{i\in\Z_r}U(1)$. As a consequence, the set of possible maps $\pi_r(\mathcal{A}_d)$ is isomorphic to 
$\mathcal{A}_d/ \bigotimes_{i\in\Z_r}U(1)$ and its dimension is given by
\begin{equation}
\dim{(\pi_r(\mathcal{A}_d))}=\dim(\mathcal{A}_d)-\dim(\bigotimes_{i\in\Z_r}U(1))
\end{equation}
Since $\mathcal{A}(d)$ is a $2d^2$ dimensional set with $2(d-1)+1$ real constraints and $\bigotimes_{i\in\Z_r}U(1)$ has $r$ dimensions. We conclude that
\begin{equation}
\dim{(\pi_r(\mathcal{A}_d))}=2d^2-(2(d-1)+1)-d=2d^2-2d+1-r
\end{equation}
\end{proof}

From this dimensional analysis, we conclude that not every UCTP qudit map with rank $r$ can be expressed as a map in $\pi_k(\mathcal{A}_d)$. This follows from the fact that a complete description of such maps would require $2d^2r-d^2-r^2$ real parameters \autocites{Iten2018, Friedland2016}. 

In the following section, we will further characterise the families of maps represented by $\pi_r(\mathcal{A}_d)$. In particular, we will find the algebraic expression in terms of $(\alpha_{ij})$ determining whether a given map $\pi_r(\alpha_{ij})$ is an extreme point of the set of UCPT maps. This property is particularly useful in the study of its convex structure.

\subsection{Quantifying the distance to the mixed-unitary set.}

The following theorem establishes whether a UCPT map as given in \eref{dfam1} corresponds to an extreme point of $\Xi^{UT}_{d,d}$.

\begin{theorem} 
An unital and trace-preserving map given by $\pi_d(\alpha_{ij})$ where $(\alpha_{ij})_{i,j\in\Z_d}\in{\mathcal{A}_d}$ as given by $\eref{matparam}$ corresponds to an extreme point of the set unital and trace-preserving maps iff the matrices $(M_l|N_l)$ are full-rank for $l=0,...,d-1$ where
\begin{eqnarray}
M_{l}=\left(\left(\sum_{j=0}^{d-1}\alpha_{i+l j}\omega^{j(k-l)}\right)\left(\sum_{j=0}^{d-1}\alpha^*_{ij}\omega^{-jk}\right)\right)_{i,k\in\Z_d}
\label{MatrixM}
\end{eqnarray}
and  
\begin{eqnarray}
N_l=\left(\left(\sum^{d-1}_{j=0}{\alpha_{i+l j}\omega^{(k-i)j}}\right)\left(\sum_{j=0}^{d-1}{\alpha^*_{i j}\omega^{-(k-i)j}}\right)\right)_{i,k\in\Z_{d}}.
\label{MatrixN}
\end{eqnarray}  
\label{EXTREME}
\end{theorem}

\begin{proof}
 
By theorem \ref{LS}, we have it that a map is an extreme point of $\Xi^{UT}_{d,d}$ if the set $\{K^{\dagger}_i K_j \oplus K_i K^{\dagger}_j\}_{i,j\in\Z_d}$ is linear independent. First, let us consider $K_i^\dagger K_{i+l}$ as
\begin{eqnarray}
K_i^\dagger K_{i+l} &= \left(\sum_{k,j=0}^{d-1}{}\alpha^*_{ij}\omega^{-kj}\qudit{k}\rqudit{k+i}\right)\left(
\sum_{j,n=0}^{d-1}{}\alpha_{i+l j}\omega^{jn}\qudit{n+i+l}\rqudit{n}\right)\nonumber\\
  &=\sum_{k,n=0}^{d-1}\left(\sum_{j=0}^{d-1}\alpha_{i+l j}\omega^{jn}\right)\left(\sum_{j=0}^{d-1}\alpha^*_{ij}\omega^{-kj}\right)\ket{k}\braket{k+i|n+i+l}\bra{n}\nonumber\\
  &=\sum_{k,n=0}^{d-1}\left(\sum_{j=0}^{d-1}\alpha_{i+l j}\omega^{j(n-l)}\right)\left(\sum_{j=0}^{d-1}\alpha^*_{ij}\omega^{-kj}\right)\ket{k}\braket{k+i|n+i}\bra{n-l}\nonumber\\
  &=\sum_{k}^{d-1}\left(\sum_{j}^{d-1}\alpha_{i+l j}\omega^{j(k-l)}\right)\left(\sum_{j}^{d-1}\alpha^*_{ij}\omega^{-j k}\right)\ket{k}\bra{k-l}.
   \label{KTK1}
\end{eqnarray}
The matrices in $\eref{KTK1}$ can be expressed in vector form assuming $\ket{k}\bra{k-l}\sim\ket{k,k-l}$ so that
\begin{eqnarray}
K_i^\dagger K_{i+l}\cong \sum_{k=0}^{d-1}{}\gamma_{ikl}\ket{k,k-l}.
\end{eqnarray}
We take the inner product of two arbitrary vectors as
\begin{eqnarray}
\braket{K_j^\dagger K_{j+n}|K_i^\dagger K_{i+l}}= \sum_{k m}^{d-1}\gamma_{ikl}\gamma^*_{jmn}\braket{m,m-n|k,k-l}.
\end{eqnarray}
We see that $\braket{K_j^\dagger K_{j+n}|K_i^\dagger K_{i+l}}=0$ if $l\neq n$ for $k,l,m,n\in\Z_{d}$. The non-zero coefficients of $\{K_i^\dagger K_{i+l}\}_{i\in\Z_d}$ can be expressed in matrix form as $M_l$ for $l=0,...,d-1$. Second, we consider $K_{i+l} K_i^\dagger$ as
\begin{eqnarray}
K_{i+l} K_{i}^\dagger &= \left(
\sum_{j,n=0}^{d-1}{}\alpha_{i+l j}\omega^{nj}\qudit{n+i+l}\rqudit{n}\right)\left(\sum_{j,k=0}^{d-1}{}\alpha^*_{ij}\omega^{-jk}\qudit{k}\rqudit{k+i}\right)\nonumber\\
  &=\sum_{n,k=0}^{d-1}\left(\sum_{j=0}^{d-1}\alpha_{i+l j}\omega^{jn}\right)\left(\sum_{j=0}^{d-1}\alpha^*_{ij}\omega^{-jk}\right)\ket{n+i+l}\braket{n|k}\bra{k+i}\nonumber\\
  &=\sum_{k=0}^{d-1}\left(\sum_{j=0}^{d-1}\alpha_{i+l j}\omega^{jk}\right)\left(\sum_{j=0}^{d-1}\alpha^*_{ij}\omega^{-jk}\right)\ket{k+i+l}\bra{k+i}\nonumber\\
  &=\sum_{k=0}^{d-1}\left(\sum_{j=0}^{d-1}\alpha_{i+l j}\omega^{(k-i)j}\right)\left(\sum_{j=0}^{d-1}\alpha^*_{ij}\omega^{-(k-i)j}\right)\ket{k+l}\bra{k}.
   \label{KTK2}
\end{eqnarray}
As we did before, we may vectorise these matrices by using $\ket{k+l}\bra{k}\sim\ket{k+l,k}$ so that
\begin{equation}
K_{i+l}K_i^\dagger\cong\sum_{k=0}^{d-1}{}\gamma_{ikl}\ket{k+l,k}.
\end{equation}
The inner product of two arbitrary vectors is expressed as
\begin{equation}
\braket{K_{j+n} K_j^\dagger|K_{i+l} K_i^\dagger}=\sum_{k m}^{d-1}{}\gamma_{ikl}\gamma^*_{jmn}\braket{m+n,m|k+l,k}
\end{equation}
and we note that $\braket{K_{j+n} K_j^\dagger|K_{i+l} K_i^\dagger}=0$ in the case that $n\neq l$. In this case, the non-zero coefficients of the sets $\{K_{i+l}K_i^\dagger\}_{i\in\Z_d}$ are given by the matrices $N_l$ with $l=0,...,d-1$. As we saw, two elements of $\{K^{\dagger}_{i+l} K_{i} \oplus K_{i}K_{i+l}^\dagger\}_{i,l\in\Z_d}$ with different $l$ are linear independent so we just require that the all the matrices $(M_{l}|N_{l})$ with $l=0,...,d-1$ to be full-rank.
\end{proof}

Despite that theorem \ref{EXTREME} is enunciated only for maps in $\pi_d(\mathcal{A}_d)$, a similar result can be established for the maps with different ranks. In this case, the matrices $(M_{l}|N_{l})$ are reformulated according to the map in $\pi_r(\mathcal{A}_d)$. In the following section, we provide examples of UCPT qutrit maps in which these matrices are explicitly evaluated for the case of rank three and rank four maps.

Following theorem \ref{EXTREME}, we may construct a measure quantifying the relation of a map in $\pi_r(\mathcal{A}_d)$ and the set of convex sums of unitary maps.

\begin{definition}
For every map given by $\zeta(\pi_r(\alpha_{ij}))$ with $\alpha_{ij}) \in \mathcal{A}_d$ define $\zeta(\pi_r(\alpha_{ij})\in \R^{+}$ as the following scalar
\begin{equation}
\zeta(\pi_r(\mathcal{A}_d)):=\frac{\min\{(\sigma(M_{l}|N_{l}))\}_{l\in\Z_r}}{\max\{(\sigma(M_{l}|N_{l}))\}_{l\in\Z_r}}
\end{equation}
where $\sigma(A)$ corresponds to the set of singular values o of $A$.
\end{definition}

By definition, the singular values of any matrix are positive so we necessarily have that $0\leq\zeta(\pi_r(\mathcal{A}_d))\leq 1$. By theorem \ref{EXTREME}, if the map $\pi_r(\alpha_{ij})$ is not UCTP-extremal, then at least one singular value is zero and we have that $\zeta(\pi_r(\alpha_{ij}))=0$. For that reason, the scalar $\zeta$ quantifies the distance between any map $\pi_r(\alpha_{ij})$ and the set of convex combinations of UCTP maps. In the following section we will consider different examples of UCTP qutrit maps and we will evaluate the scalar $\zeta$ for those maps. 

\section{Qutrit maps}\label{examples}

Theorem \ref{LS} establishes that $\sqrt{18}$ is an upper bound for the rank of  UCPT-extremal qutrit maps. We may then consider four different classes of UCPT-extremal qutrit maps based on their rank. All rank one UCTP maps correspond to unitary maps, $\mathcal{E}(\rho)=U \rho U$ with $U^{\dagger}U=\mathbb{1}$. For such maps, the UCPT-extremal conditions given by theorem \ref{LS} are trivially satisfied. It is the case that all rank two qutrit maps admit a decomposition in terms of other UCTP qutrit maps \autocite{Landau1993} and consequently we cannot find rank two UCPT-extremal qutrit maps. The rest of UCPT-extremal qutrit maps are either of rank three or four.

In this section, we consider different examples of UCPT-extremal qutrit maps through the parametrised families of maps derived in the previous section. In particular, 
we will construct different matrices in $\mathcal{A}_3$ and we will obtain their associated quantum maps through the mappings given by $\pi_3$ and $\pi_4$. For such maps, we will consider their relationship with respect to the set of UCPT-extremal maps by using the scalar $\zeta$.

\subsection{Examples}
Consider the matrix $E(a,b)\in\mathcal{A}_3$ defined as the $3\times 3$ matrix with $1$ in the $(a,b)$th entry and $0$s elsewhere. In total, we have nine possible quantum maps $\{\pi_3(E(a,b))\}_{a,b\in\Z_3}$ which transform $\rho$ as
\begin{equation}
\pi_3(E(a,b))(\rho)=X_aZ_b\rho (X_aZ_b)^{\dagger} \quad\textrm{for}\quad a,b\in\Z_3.
\label{firstex}
\end{equation}
The maps $\pi_3(E(a,b))$ are all unitary (rank one) and orthogonal between them. The set $\{\pi_3(E(a,b))\}_{a,b\in\Z_3}$ corresponds to the set of Heisenberg-Weyl channels of dimension three which is analogous to the set of Pauli channels for dimension two. For any of such maps, we have that $\zeta(\pi_3(E(a,b)))=1$ which is the maximum value of $\zeta$ that any map can achieve. 

Consider now the matrix $A\in\mathcal{A}_3$ which is given by
\begin{equation}
A=\frac{\sqrt{2}}{6}\left(\begin{array}{ccc}
0 & 1-\omega^2 & 1-\omega\\
0 & \omega^2-\omega & \omega-\omega^2\\
0 & \omega-1 & \omega^2-1
\end{array}\right).
\label{par3}
\end{equation}
The quantum map given by $\pi_3(A)$ acts on a density matrix $\rho$ as $\pi_3(A)(\rho)=\sum_{i=0}^{2}\frac{1}{\sqrt{2}}K_i\rho K_i^{\dagger}$ where
\begin{equation}
\fl \{K_0,K_1,K_2\}=\left\{
\left(\begin{array}{ccc}
1&0&0\\
0&-1&0\\
0&0&0
\end{array}\right),
\left(\begin{array}{ccc}
0&0&-1 \\
0&0&0  \\
0&1&0  \\
\end{array}\right),
\left(\begin{array}{ccc}
0&0&0  \\
0&0&1   \\
-1&0&0  \\
\end{array}\right)\right\}.
\end{equation}
This quantum map may result not be familiar at first glance, however, we can show that $\pi_3(A)$ is equivalent up to a unitary rotation to the well-known anti-symmetric Werner–Holevo (ASWH) map over dimension three. In particular, we have that the ASWH map can be expressed as $\mathcal{E}_{ASWH}(\rho)=U\pi_3(A)(\rho) U^{\dagger}$ with
\begin{equation}
U=\left(\begin{array}{ccc}
0&1&0  \\
1&0&0  \\
0&0&1  
\end{array}\right).
\end{equation}
One peculiarity of the ASWH map (and consequently the map $\pi_3(A)$) is that it maximises the distance with respect to the set of all convex sums of unitaries \cite{Mendl2008}. We may now evaluate the scalar $\zeta(\pi_3(A))$ in this case. For this map, we obtain that
\begin{equation}
\left(M_{0}\vert N_{0}\right)=
\frac{1}{2}\left(\begin{array}{ccc|ccc}
0 & 1 & 1 & 1 & 0 & 1\\
1 & 1 & 0 & 1 & 1 & 0\\
1 & 0 & 1 & 0 & 1 & 1
\end{array}\right)
\end{equation} and
\begin{equation}
\left(M_{1}\vert N_{1}\right)=
\frac{1}{2}\left(\begin{array}{ccc|ccc}
0 & -1 & 0 & -1 & 0 & 0\\
-1 & 0 & 0 & 0 & -1 & 0\\
0 & 0 & -1 & 0 & 0 & -1
\end{array}\right).
\end{equation} We get the singular values $\sigma\left(M_{0}\vert N_{0}\right)=\{\sqrt{2},\frac{1}{\sqrt{2}},\frac{1}{\sqrt{2}}\}$ and $\sigma\left(M_{1}\vert N_{1}\right)=\{\frac{1}{\sqrt{2}},\frac{1}{\sqrt{2}},\frac{1}{\sqrt{2}}\}$, respectively. So we conclude that 
\begin{equation}
\zeta(\pi_3(A))=\frac{1}{2}.
\end{equation}

Consider now the map given by $\pi_4(A)$. In terms of the operator-sum representation, this map can be expressed as $\pi_4(A)(\rho)=\sum_{i=0}^{4}\frac{1}{\sqrt{2}}K_i\rho K^{\dagger}_i$ where
\begin{equation}
\fl
\{K_0,K_1,K_2,K_3\}=\left\{
\left(\begin{array}{ccc}
1&0&0\\
0&-1&0\\
0&0&0
\end{array}\right),
\left(\begin{array}{ccc}
0&0&0 \\
0&0&0 \\
0&1&0 \\
\end{array}\right),
\left(\begin{array}{ccc}
0&0&-1 \\
0&0&0  \\
-1&0&0  \\
\end{array}\right),
\left(\begin{array}{ccc}
0&0&0  \\
0&0&1   \\
0&0&0  \\
\end{array}\right)\right\}.
\end{equation}
We wish to evaluate $\zeta(\pi_4(A))$. To do so we first need to evaluate the singular values of 
\begin{equation}
\left(M'_{0}\vert N'_{0}\right)=\frac{1}{2}
\left(\begin{array}{ccc|ccc}
1 & 1 & 0 & 1 & 1 & 0\\
0 & 1 & 0 & 0 & 0 & 1\\
1 & 0 & 1 & 1 & 0 & 1\\
0 & 0 & 1 & 0 & 1 & 0
\end{array}\right),
\end{equation}
\begin{equation}
\left(M'_{1}\vert N'_{1}\right)=\frac{1}{2}
\left(\begin{array}{cc|cc}
0 & 0 & 0 & -1 \\
-1 & 0 & 0 & 0 \\
0 & 0 & -1 & 0 \\
0 & -1 & 0 & 0 
\end{array}\right)
\end{equation}
\noindent
and
\begin{equation}
\left(M'_{2}\vert N'_{2}\right)=\frac{1}{2}
\left(\begin{array}{cc|cc}
-1 & 0 & 0 & -1 \\
0 & 0 & 0 & 0 \\
0 & -1 & -1 & 0 \\
0 & 0 & 0 & 0 
\end{array}\right)
\end{equation}
which are given by $\sigma\left(M'_{0}\vert N'_{0}\right)=\{\sqrt{1+\frac{1}{\sqrt{2}}},\frac{1}{\sqrt{2}},\frac{1}{\sqrt{2}},\sqrt{1-\frac{1}{\sqrt{2}}}\}$ and $\sigma\left(M_{1}\vert N_{1}\right)=\{\frac{1}{\sqrt{2}},\frac{1}{\sqrt{2}},\frac{1}{\sqrt{2}}\}$. We conclude that this map
\begin{equation}
\zeta(\pi_4(A))=0
\end{equation}
or, equivalently, $\pi_4(A)$ can be expressed in terms of a convex decomposition of other UCPT maps. The family of maps given by $\pi_4(\mathcal{A}_3)$ includes also UCTP-extremal maps. To see this, we consider a different example. In this case, consider the matrix $B\in\mathcal{A}_3$ given by
\begin{equation}
B=\frac{1}{\sqrt{7}}\left(\begin{array}{ccc}
0 & 1 & -1\\
0 & 1 & 1+\omega\\
1 & 1+\omega^2 & \omega^2
\end{array}\right).
\label{par3}
\end{equation}The rank four UCTP map $\pi_4(B)$ can be expressed in terms of the operator-sum representation as $\pi_4(B)(\rho)=\sum_{i=0}^{4}\frac{1}{\sqrt{7}}K_i\rho K^{\dagger}_i$ where
\begin{eqnarray}
\fl
\{K_0,K_1,K_2,K_3\}=
\nonumber
\\
\fl
\left\{
\left(\begin{array}{ccc}
0&0&0\\
0&\scriptstyle{\sqrt{3}}\ramuno&0\\
0&0&\scriptstyle{-\sqrt{3}}\ramuno
\end{array}\right),
\left(\begin{array}{ccc}
0&0&0 \\
\scriptstyle{1+\omega}&0&0 \\
0&0&0 \\
\end{array}\right),
\left(\begin{array}{ccc}
0&0&\scriptstyle{-1-\omega} \\
0&0&0 \\
\scriptstyle{1-\sqrt{3}\ramuno} & 0 & 0  \\
\end{array}\right),
\left(\begin{array}{ccc}
0& \scriptstyle{1+\sqrt{3}}\ramuno&0  \\
0&0&\scriptstyle{1} \\
0&0&0  \\
\end{array}\right)\right\}.
\end{eqnarray}

Again, we are interested in $\zeta(\pi_4(B))$. Then, we need to evaluate the singular values of the matrices
\begin{equation}
\left(M'_{0}\vert N'_{0}\right)=\frac{1}{7}
\left(\begin{array}{ccc|ccc}
0 & 3 & 3 & 0 & 3 & 3\\
3 & 0 & 0 & 3 & 0 & 0\\
4 & 0 & 3 & 3 & 0 & 4\\
0 & 4 & 1 & 4 & 1 & 0
\end{array}\right),
\end{equation}
\begin{equation}
\left(M'_{1}\vert N'_{1}\right)=\frac{1}{7}
\left(\begin{array}{cc|cc}
3-\sqrt{3}\ramuno & 1-\omega^2 & 0 & 1+\omega \\
1+\omega & 0 & 0 & 2\sqrt{3}\ramuno \\
0 & -2\sqrt{3}\ramuno & 0 & 0 \\
0 & 0 & 3 & 0 
\end{array}\right)
\end{equation}
\noindent
and
\begin{equation}
\left(M'_{2}\vert N'_{2}\right)=\frac{1}{7}
\left(\begin{array}{cc|cc}
0 & 3 & 2+2\omega & -1 \\
0 & 0 & 1-\omega^2 & 0 \\
3 & 0 & 0 & 3-\sqrt{3}\ramuno \\
0 & 0 & 0 & 1+\omega 
\end{array}\right).
\end{equation}

We obtain that their singular values are given by $\sigma\left(M'_{0}\vert N'_{0}\right)=\{1.28,0.749,0.600,0.498\}$, $\sigma\left(M'_{1}\vert N'_{1}\right)=\{0.779,0.508,0.429,0.230\}$, $\sigma\left(M'_{2}\vert N'_{2}\right)=\{0.682,0.682,0.155,0.155\}$. For this map, we obtain that
\begin{equation}
\zeta(\pi_4(B))=\frac{0.155}{1.28}=0.121.
\end{equation}
This proves that the family of maps given by $\pi_4(A_3)$ includes also UCPT-extremal maps being and represents a possible ansatz to investigate qutrit maps of rank four different from the provided in previous works \cite{Mendl2008,haagerup2021extreme}.

\section{Conclusions}

In this paper, we considered the characterisation of the convex structure of the set of unital quantum channels. To do this we proposed a novel framework for the study of such maps. Firstly, we considered a particular family of quantum maps admitting a Kraus decomposition in terms of linear combinations of Heisenberg-Weyl operators. Secondly, we introduced a parametrisation for this family of maps which we use to derive the equations required to satisfy the unital and trace-preserving conditions. Inspired by this parametrization, we proposed a generalization of the family of quantum maps, which allows considering other families of maps with different ranks. Finally, we proposed a measure of the distance between the generalised family of maps and the set of mixed-unitary maps.

As an application of this framework, we considered the particular case of UCPT maps over dimension three. Our framework is especially useful in this set-up as it allows us to consider different non-trivial extremal points of the set. To see that, we constructed examples of UCPT qutrit maps of ranks three and four. We showed that some of those examples can be associated with well-known maps in the literature such as the Heisenberg-Weyl channels or the ASWH qutrit map. We considered the relation between the examples of maps presented and the set of mixed-unitary maps in terms of the measure we defined. Our framework also provides a more comprehensive alternative to the construction of qutrit maps with rank four appearing in \cite{Mendl2008}.

In a conclusion, after considering a partial parametrisation of the set of unital and trace-preserving maps, we believe that a complete parametrisation of UCTP qutrit maps is attainable. We believe that our work could open the way for a better comprehension of the set. For future work, we will consider the features of the bipartite states associated with the families of maps presented. In particular, we will consider the application of the Choi-Jamio\l kowski isomorphism between such bipartite states and our maps.

\newpage
\printbibliography

\end{document}